\documentclass[11pt,letterpaper]{article}
 \usepackage{fullpage}
 \usepackage[T1]{fontenc}
 \usepackage{lmodern}
 
 \usepackage[utf8]{inputenc}
 \usepackage[english]{babel}
 \usepackage{tikz,amsmath,amssymb,amsfonts,latexsym,graphicx,amsthm,algorithmicx}
 \usepackage{subcaption}
 \usepackage{bbm}
 \usepackage{comment}
 \usepackage{thmtools}
 \usepackage{thm-restate}
 \usepackage{comment}
 \declaretheorem{theorem}
 \usepackage{algorithm}
 \usepackage[noend]{algpseudocode}
 \usepackage{hyperref}
 \usepackage[capitalise]{cleveref}
 \theoremstyle{plain}
 \newtheorem{lemma}[theorem]{Lemma}
 
 \newtheorem{fact}[theorem]{Fact}
 
 \newtheorem{remark}[theorem]{Remark}
 \theoremstyle{definition}
 \newtheorem{definition}[theorem]{Definition}

\newcommand{\OPT}{\mathrm{OPT}}

\newcommand{\opt}{\mathrm{opt}}

\newcommand{\dist}{\mathrm{dist}}
\newcommand{\length}{\mathrm{length}}

\newcommand{\eps}{\varepsilon}
\newcommand{\biggamma}{\Gamma}

\title{An Approximation Algorithm for \\ Distance-Constrained Vehicle Routing on Trees}
\date{}
\author{Marc Dufay\thanks{IRIF and Ecole Polytechnique, France, e-mail: \texttt{marc.dufay@polytechnique.edu}} \and Claire Mathieu\thanks{CNRS Paris, France, e-mail: \texttt{claire.mathieu@irif.fr}.}  \and Hang Zhou\thanks{Ecole Polytechnique, France, e-mail: \texttt{hzhou@lix.polytechnique.fr}.}}

\usepackage{xcolor}
\usepackage[strict]{changepage}
\usepackage{framed}
\definecolor{formalshade}{rgb}{0.95,0.95,1}
\definecolor{darkblue}{rgb}{0.0, 0.0, 0.55}
\newenvironment{myquote}{%http://www.jevon.org/wiki/Fancy_Quotation_Boxes_in_Latex
  \MakeFramed{\advance\hsize-\width\FrameRestore}%
  \noindent\hspace{-4.55pt}% disable indenting first paragraph
  \begin{adjustwidth}{}{7pt}%
  \vspace{2pt}\vspace{2pt}%
}
{%
  \vspace{2pt}\end{adjustwidth}\endMakeFramed%
}

\begin{document}

\maketitle

\begin{abstract}

In the Distance-constrained Vehicle Routing Problem (DVRP), we are given a  graph with integer edge weights, a depot, a set of $n$ terminals, and a distance constraint $D$. The goal is to find a minimum number of tours starting and ending at the depot such that those tours together cover all the terminals and the length of each tour is at most $D$.

%The DVRP is NP-hard even on trees.
The DVRP on \emph{trees} is of independent interest, because it is equivalent to the ``virtual machine packing'' problem on trees studied by Sindelar et al. [SPAA'11].
We design a simple and natural approximation algorithm for the tree DVRP, parameterized by $\varepsilon >0$. We show that its approximation ratio is $\alpha + \varepsilon$, where $\alpha \approx 1.691$, and in addition, that our analysis is essentially tight. The running time  is polynomial in $n$ and $D$. The  approximation ratio improves on the ratio of 2 due to Nagarajan and Ravi [Networks'12].

The main novelty of this paper lies in the analysis of the algorithm. It relies on a reduction from the tree DVRP to the bounded space online bin packing problem via a new notion of ``reduced length''.
\end{abstract}

\newpage{}
\setcounter{page}{1}

\section{Introduction}
The vehicle routing problem is arguably one of the most important problems in Operations Research.
Books have been dedicated to vehicle routing problems, e.g.,\ \cite{crainic2012fleet,golden2008vehicle,toth2002vehicle}.
Yet, these problems remain challenging, both from a practical and a theoretical perspective.
As observed by Li, Simchi-Levi, and Desrochers~\cite{li1992distance}:
\begin{myquote}
Typically, two types of constraints are imposed on the route traveled by any vehicle. One is the \emph{capacity constraint} in which each vehicle cannot serve more than a given number of customers.
The second is the \emph{distance constraint}: The total distance traveled by each vehicle should not exceed a prespecified number.
Depending on the system, one or both types can be binding.
Usually delivery and pick-up services are characterized by capacity constraints [\dots], while service systems are characterized by distance constraints (see, for example, \cite{assad}). On the latter case, the system is usually required to provide a visit of a skilled worker at customer sites and thus the length of routes must be controlled because these are related to working hours.
\end{myquote}

%The vehicle routing with a capacity constraint and when the objective is to minimize the total distance traveled has been studied extensively on general metrics~\cite{haimovich1985bounds,altinkemer1987heuristics,altinkemer1990heuristics,bompadre2006improved,blauth2021improving,friggstad2021improved} as well as in the special case when the graph is a tree~\cite{labbe1991capacitated,hamaguchi1998capacitated,asano2001new,becker2018tight,becker2019framework,jayaprakash2021approximation,MZ22a}.

%The vehicle routing with a distance constraint, called \emph{Distance-constrained Vehicle Routing Problem (DVRP)}, has many practical applications, see, e.g., ~\cite{li1992distance,assad,laporte1984two}.

The focus of this paper is the distance constraint.
In the \emph{Distance-constrained Vehicle Routing Problem (DVRP)},
%see, e.g.,~\cite{li1992distance,assad,laporte1984two,nagarajan2012approximation,friggstad2014approximation,becker2019framework}.
we are given a graph with integer edge weights, a vertex called \emph{depot}, a set of $n$ vertices called \emph{terminals}, and an integer \emph{distance constraint} $D$.
%,   which is polynomially bounded in the size of the graph.
The goal is to find a minimum number of tours starting and ending at the depot such that those tours together cover all the terminals and the length of each tour is at most $D$.
Friggstad and Swamy~\cite{friggstad2014approximation} gave an $O(\frac{\log D}{\log \log D})$-approximation, improving upon an $O(\log D)$-approximation of Nagarajan and Ravi~\cite{nagarajan2012approximation}.\footnote{More precisely, Nagarajan and Ravi~\cite{nagarajan2012approximation} designed a bicriteria $\left(O(\log \frac{1}{\eps}), 1+\eps\right)$-approximation, which could be turned into an $O(\log D)$-approximation by setting $D=\frac{1}{\eps}$.}
Experimental results were given in \cite{laporte1984two}.

The DVRP on \emph{trees} is of independent interest, because of its relation to the \emph{Virtual Machine (VM) packing problem}~\cite{sindelar2011sharing,barker2012empirical,rampersaud2016sharing}.
In the VM packing problem, we are given a set of VMs that must be hosted on physical servers, where each VM consists of a set of pages and each physical server has a capacity of $D$ pages.
VMs running on the same physical server may share pages.
The goal is to pack the VMs onto the smallest number of physical servers.
Sindelar et al.~\cite{sindelar2011sharing} observed:
\begin{myquote}
Using memory traces for a mixture of diverse OSes, architectures, and software libraries, we find that \emph{a tree model} can capture up to 67\% of inter-VM sharing from these traces.
\end{myquote}
\noindent Sindelar et al.~\cite{sindelar2011sharing} gave a 3-approximation for the VM packing problem on trees, and also suggested as future work ``\emph{A key direction is tightening the approximation bounds}''.

It is easy to see that the VM packing problem on trees is equivalent to the DVRP on trees.
Thus the algorithm of Sindelar et al.~\cite{sindelar2011sharing} is a 3-approximation for the tree DVRP.
Nagarajan and Ravi~\cite{nagarajan2012approximation} improved the ratio of the tree DVRP to $2$.
When the distance bound is allowed to be violated by an $\eps$ fraction, Becker and Paul~\cite{becker2019framework} designed a bicriteria PTAS.
We observe that the tree DVRP is strongly NP-hard (\cref{sec:NP-hard}).
%by a reduction from the strongly NP-hard bin-packing problem.
%It is open whether there is a true PTAS for tree DVRP.

In this work, we design a simple and natural approximation algorithm for the tree DVRP, parameterized by $\epsilon >0$, see \cref{alg:main}.
The main novelty lies in the analysis of \cref{alg:main}.
Our main result (\cref{thm:alpha-approx}) shows that the approximation ratio of \cref{alg:main} is $\alpha + O(\epsilon)$, where $\alpha \approx 1.691$ is defined in \cref{def:alpha}.
The running time is polynomial in $n$ and $D$.
Interestingly, the ratio $\alpha$ is best possible for \cref{alg:main} (\cref{thm:lower-bound}).

%Our algorithm (\cref{alg:main}) is essentially a reduction to instances for which the optimal number of tours is $O(1)$.
%That special case can be solved exactly (\cref{lem:exact-algo}).

\begin{definition}
\label{def:alpha}
Let $(u_k)_{k\geq 1}$ denote the following sequence:
\[u_k=\left\{
\begin{array}{lll}
1,   & k=1,\\
u_{k-1} ( u_{k-1} + 1), & k\geq 2.
\end{array}
\right. \]
\[
\text{ Let } \alpha := \sum_{k=1}^{+\infty} \frac{1}{u_k} = 1.69103\dots.
\]
\end{definition}

\begin{theorem}
\label{thm:alpha-approx}
For any constant $\eps>0$, \cref{alg:main} is an $(\alpha + O(\epsilon))$-approximation  for the \emph{Distance-constrained Vehicle Routing Problem (DVRP)} on trees.
Its running time is $O\left(n^2\cdot \left( D / \epsilon^2 \right)^{O(1/\epsilon^2)}\right)$.
\end{theorem}

\begin{remark}
It is common in vehicle routing to parameterize the running time of an algorithm by the value of a constraint. For example, in capacitated vehicle routing with splittable demands, the running time of the algorithms in~\cite{jayaprakash2021approximation,MZ22a} is parameterized by the \emph{tour capacity}.
\end{remark}

\begin{theorem}
\label{thm:lower-bound}
For any constant $\eps>0$,
%For any $\mu > 0$, there exists an instance $\mathcal{I}$ such that the approximation ratio of Algorithm~\ref{alg:main} on $\mathcal{I}$ is at least $\alpha-\mu$.
\cref{alg:main} is at best an $\alpha$-approximation for the \emph{Distance-constrained Vehicle Routing Problem (DVRP)} on trees.
\end{theorem}

It is an open question to achieve a better-than-$\alpha$ approximation for the DVRP on trees.
From \cref{thm:lower-bound}, this would require new insights in the algorithmic design.

\subsection{Algorithm}
\label{sec:main-algo}
Let $\eps > 0$.
 Our algorithm for the DVRP is \cref{alg:main}.
 It consists of two phases, using \cref{alg:components} and \cref{alg:local}  with  $\biggamma=1/\eps^2$:
\begin{description}
    \item[Phase 1:] The tree is partitioned into \emph{components} (\cref{lem:comp_dec}  and \cref{alg:components}), where each component can be covered with a bounded  number of tours.
    \item[Phase 2:] Each component is taken as an independent instance for the DVRP, which is solved using a polynomial time dynamic program (\cref{lem:exact-algo} and \cref{alg:local}); the solution for the whole tree is the union of the solutions for  individual components.
\end{description}

\begin{algorithm}[h]
\caption{Approximation algorithm for the DVRP on trees. Parameter $\eps >0$.}
\label{alg:main}
\hspace*{\algorithmicindent} \textbf{Input:} A tree $T$ rooted at $r$, a set of terminals $U$, a distance constraint $D$ \\
\hspace*{\algorithmicindent} \textbf{Output:} number of tours in a feasible solution to cover all terminals in $U$
\begin{algorithmic}[1]
\State $\Gamma \gets 1/\epsilon^2$
\State Partition the tree $T$ into a set $\mathcal{C}$ of components
\Comment{\cref{alg:components}}
\For{each component $c\in \mathcal{C}$}
    \State $r_c\gets$ root vertex of component $c$ \Comment{defined in \cref{lem:comp_dec}}
    \State $D_c\gets D$ minus twice the distance between $r$ and $r_c$
    \State $U_c\gets$ set of terminals in $U$ that belong to $c$
    \State $n_c\gets$ minimum number of tours for the subproblem $(c,U_c,D_c)$\Comment{\cref{alg:local}}
\EndFor
\State \Return $\sum_{c\in \mathcal{C}} n_c$
\end{algorithmic}
\end{algorithm}

\begin{remark}
As written, \cref{alg:main} returns the number of tours, not the tours themselves. If desired, by adding auxiliary information to the dynamic program of \cref{alg:local} in a standard manner, it is possible to retrieve a feasible solution whose number of tours matches the value returned by \cref{alg:main}.
\end{remark}

\subsection{Overview of the Analysis}
\label{sec:analysis}
This section gives a high-level description of main new ideas in this paper.

The analysis of \cref{alg:main} relies on a reduction (\cref{thm:reduction}) from the DVRP on trees to the \emph{bounded space online bin packing} (\cref{def:online-packing}).

What does bin packing have to do with DVRP on trees?
The relation lies in a new notion introduced in this paper, that of \emph{reduced lengths} (\cref{def:reduced-weight}).
If we consider a bin in bin packing, its item sizes sum to at most 1. Similarly, if we consider a tour in DVRP on trees, the reduced lengths of its subtours in components sum to at most 1 (\cref{thm:comp_sub}).

To show the reduction (\cref{thm:reduction}), we start from an instance of the tree DVRP and we construct an instance of the bounded space online bin packing as follows.
Consider the reduced lengths of all subtours in an optimal solution.
We construct an online sequence of those reduced lengths such that reduced lengths of the subtours in the same component are \emph{consecutive} in the sequence.
Then we consider a solution to the bounded space online bin packing on that sequence.
Intuitively, the bound on the number of open bins implies that bins in that solution tend to contain only reduced lengths of the subtours from the same component. That is a desirable property because, if a bin contains only  reduced lengths of subtours in the same component, then those subtours can be combined into a single tour (\cref{lemma:single-tour}).
Using the above intuition, we show that the performance of \cref{alg:main} for the tree DVRP is up to some negligible additive factor at least as good as the best performance for the bounded space online bin-packing problem.

From the reduction (\cref{thm:reduction}) and since the bounded space online bin packing admits an $\alpha$-competitive algorithm due to Lee and Lee~\cite{leeleeonlinepacking}, we conclude that  \cref{alg:main} is an $(\alpha+O(\eps))$ approximation for the DVRP on trees (\cref{thm:alpha-approx}).

There are several technical details along the way. For example, subtours that end in a component and subtours that traverse a component cannot be combined in the same way, which requires additional care in the definition of reduced lengths, see \cref{fig:reduced_length}.

Finally, we show that the analysis in \cref{thm:alpha-approx} on the approximation ratio of \cref{alg:main}  is essentially tight by providing a matching lower bound (\cref{thm:lower-bound}).

\subsection{Organization of the Paper}
In \cref{sec:preli}, we give the formal problem definition, preliminary results, and previous techniques.
In \cref{sec:reduced}, we define and analyze \emph{reduced lengths}.
In \cref{sec:proof-alpha-approx}, we prove \cref{thm:alpha-approx} by establishing the reduction (\cref{thm:reduction}).
In \cref{sec:lower-bound}, we prove \cref{thm:lower-bound}.

\section{Preliminaries}
\label{sec:preli}
\subsection{Formal Problem Definition, Notations, and Assumptions}
\label{sec:problem}
Let $T$ be a rooted tree $(V,E)$ with non-negative integer edge weights $w(u,v)$ for all $(u,v)\in E$.
Consider a tour (resp.\ subtour) $t=(v_0,v_1,v_2,\dots,v_m)$ for some $m\in \mathbb{N}$ such that $v_0=v_m$.
The \emph{length} of $t$, denoted by $\length(t)$, is defined to be $\sum_{i=1}^{m} w(e_i)$, where $e_i$ is the edge between $v_{i-1}$ and $v_{i}$.

%We say that a tour \emph{visits} a terminal if it is part of the tour.
\begin{definition}[tree DVRP]
\label{def:tree-dvrp}
An instance $(T,U,D)$ of the \emph{Distance-constrained Vehicle Routing Problem (DVRP)} on \emph{trees} consists of
\begin{itemize}
    \item a rooted \emph{tree} $T=(V,E)$ with  non-negative integer edge weights, where the {root} $r\in V$ of the tree is  the \emph{depot};
    \item
    a set $U\subseteq V$ of $n$ vertices of the tree $T$, called \emph{terminals};
    \item a positive integer \emph{distance constraint} $D$.
\end{itemize}
A feasible solution is a set of tours such that
\begin{itemize}
    \item each tour starts and ends at $r$;
    \item each terminal is visited by at least one tour;
    \item each tour has length at most $D$.
\end{itemize}
The goal is to minimize the total number of tours in a feasible solution.
\end{definition}
Let $\OPT$ denote an optimal solution to the tree DVRP.
Let $\opt$ denote the number of tours in $\OPT$.

For any vertices $u,v\in V$, let $\dist(u,v)$ denote the \emph{distance} between $u$ and $v$ in the tree $T$.

Up to a preprocessing step, we can assume that each vertex has at most two children and that the terminals are the same as the leaves of the tree, see, e.g.,~\cite{MZ22a}. Furthermore, we assume that each non-leaf vertex has exactly two children.\footnote{Indeed, if a vertex $u$ has only one child $v$, there are two cases. In the first case, $u$ is the root. Then we remove $u$ and its incident edge, let $v$ be the depot, and update $D$ by $D-2 w(u,v)$. In the second case, $u$ has a parent $p$. Then we remove $u$ and its incident edges, and we add an edge between $p$ and $v$ with weight $w(p,v):=w(p,u) + w(u,v)$.}
We assume that  for each tour in the solution, each edge on that tour is traversed exactly twice, once in each direction.
We assume w.l.o.g. that $\eps > 0$ is upper bounded by a sufficiently small constant.

\subsection{Decomposition into Components}
We decompose the tree into components in \cref{lem:comp_dec}.

\begin{lemma}[\cite{MZ22a}]
\label{lem:comp_dec}
Let $\biggamma\geq 1$ be a fixed integer.
There is an algorithm \textsc{Decompose}$^{(\Gamma)}$ (Algorithm~\ref{alg:components}) that runs in time $O(n^2\cdot (2\Gamma)^{\Gamma}\cdot D^{3\Gamma})$ and computes
a partition of the edges of the tree $T$ into a set $\mathcal{C}$ of \emph{components} (see \cref{fig:component}), such that all of the following properties are satisfied:
\begin{enumerate}
    \item Every component $c\in \mathcal{C}$ is a connected subgraph of $T$;
    the \emph{root} vertex of the component $c$, denoted by $r_c$, is the vertex in $c$ that is closest to the depot.
    \item We say that a component $c\in \mathcal{C}$ is a \emph{leaf} component if all descendants of $r_c$ in tree $T$ are in $c$, and
    an \emph{internal} component otherwise.
    An internal component $c$ shares vertices with other components at two vertices only: at vertex $r_c$, and at one other vertex, called the \emph{exit} vertex of the component $c$, and denoted by $e_c$.
    \item The terminals of each component can be covered by at most $\biggamma$ tours.
    We say that a component is \emph{big} if at least $\biggamma/2$ tours are needed to cover its terminals.
    Each leaf component is big.
    \item If the number of components in $\mathcal{C}$ is strictly greater than 1, then there exists a map from all components to big components, such that each big component has at most three pre-images.
\end{enumerate}
\end{lemma}

\begin{figure}[t]
    \centering
    \includegraphics[scale=0.4]{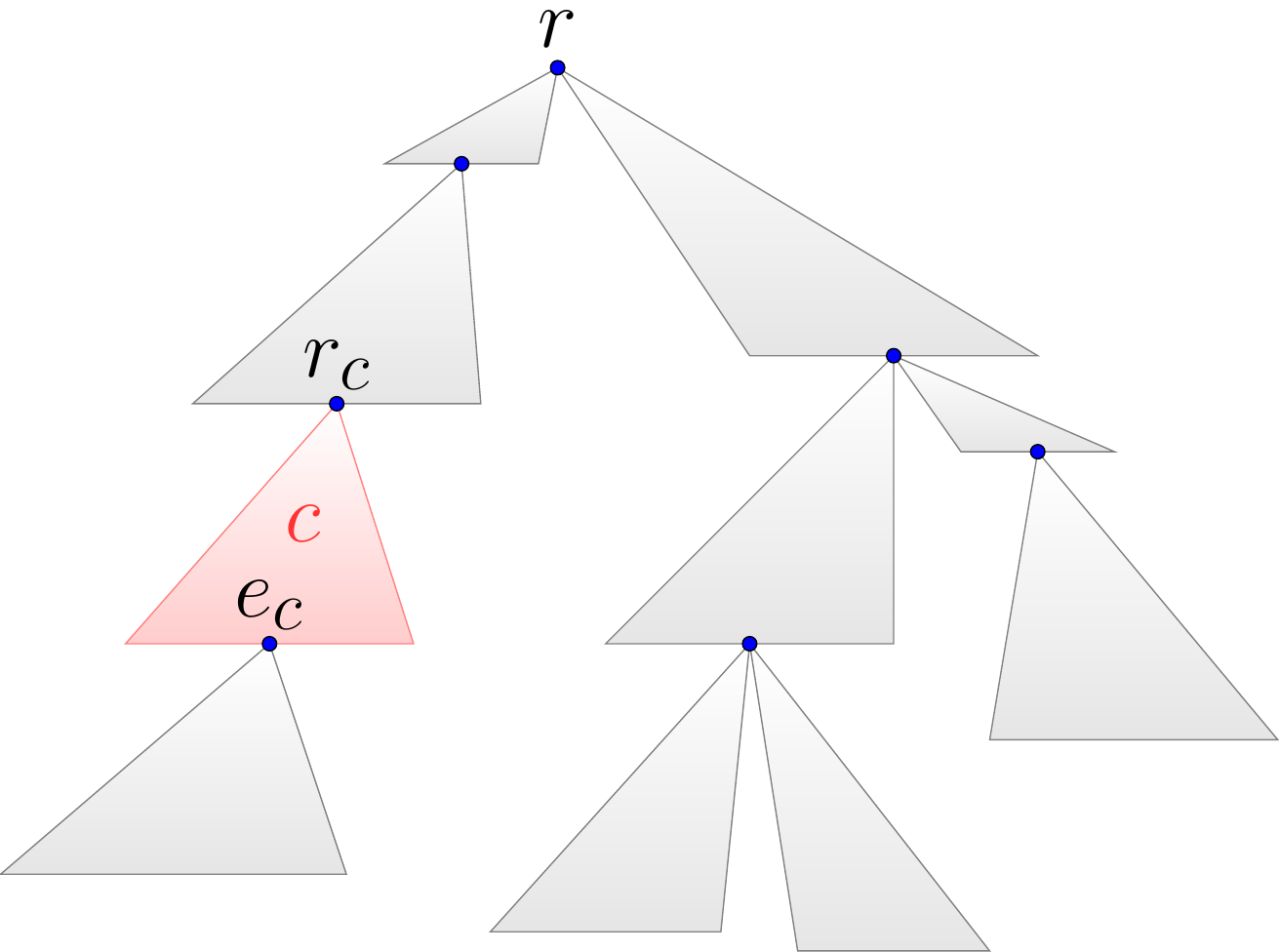}
    \caption{Component decomposition. Extracted from \cite{MZ22a}.}
    \label{fig:component}
\end{figure}

The component decomposition was previously given in \cite{MZ22a}, inspired by Becker and Paul~\cite{becker2019framework}.
Algorithm~\ref{alg:components} is a small adaptation from \cite{MZ22a}, and is given in \cref{sec:alg-component}.
The proof of \cref{lem:comp_dec} is almost identical to that in \cite{MZ22a}, hence omitted.

\begin{definition}[subtours and their categories, \cite{MZ22a}]
    Let $c\in \mathcal{C}$ be any component. A \emph{subtour} in  component $c$ is a walk inside $c$ that starts and ends at $r_c$ and visits at least one terminal.
    For any subtour $s$ in component $c$, we say that the \emph{category} of $s$ is \emph{passing} if $c$ is an internal component and the exit vertex $e_c$ belongs to $s$, and \emph{ending} otherwise.
\end{definition}

\subsection{Solving Instances  with a Bounded Number of Tours}
For an instance of the tree DVRP admitting a solution of a bounded number of tours, an optimal solution can be computed in polynomial time using a simple dynamic program (\cref{alg:local}), in \cref{lem:exact-algo}.

\begin{lemma}
\label{lem:exact-algo}
Let $\Gamma\geq 1$ be a fixed integer.
There is an algorithm \textsc{Solve}$^{(\Gamma)}$ (\cref{alg:local}) that, given an instance $(\tilde T, \tilde U, \tilde D)$ of the DVRP on trees, computes \[\min \{ |S|: S \text{ is a feasible set of tours and }|S|\leq \Gamma \}.\] Thus \textsc{Solve}$^{(\Gamma)}$ returns $+\infty$ if the instance does not admit any solution of at most $\Gamma$ tours.
The running time of \textsc{Solve}$^{(\Gamma)}$ is
%\md{$\mathcal{O}\left( \tilde n\cdot \Gamma \cdot \tilde D^{2\Gamma} \right)$},
$O(\tilde n\cdot (2\Gamma)^{\Gamma}\cdot  \tilde D^{3\Gamma})$,
where $\tilde n$ is the number of terminals in $\tilde T$.
\end{lemma}
Algorithm~\ref{alg:local} and the proof of \cref{lem:exact-algo} are in \cref{sec:proof-exact-algo}.

%----------------------------------------------

\section{Reduced Lengths}
\label{sec:reduced}
In this section, we introduce a novel concept of \emph{reduced lengths} (\cref{def:reduced-weight}).

\begin{definition}[reduced lengths, see \cref{fig:reduced_length}]
\label{def:reduced-weight}
Let $c \in \mathcal{C}$ be any component.
For any subtour~$s$ in component $c$, we define the \emph{reduced length} $\bar{\ell}(s)$ of subtour $s$ as follows:
\begin{itemize}
\item If $s$ is an \emph{ending subtour},
\[\bar{\ell}(s):=\frac{\length(s)}{D-2\cdot\dist(r,r_c)};\]
\item If $s$ is a \emph{passing subtour},
\[\bar{\ell}(s):=\frac{\length(s)-2\cdot\dist(r_c,e_c)}{D-2\cdot\dist(r,e_c)}.\]
\end{itemize}
\end{definition}

\begin{figure}[h]
\begin{subfigure}[t]{0.5\textwidth}
\centering
\includegraphics[scale=0.5]{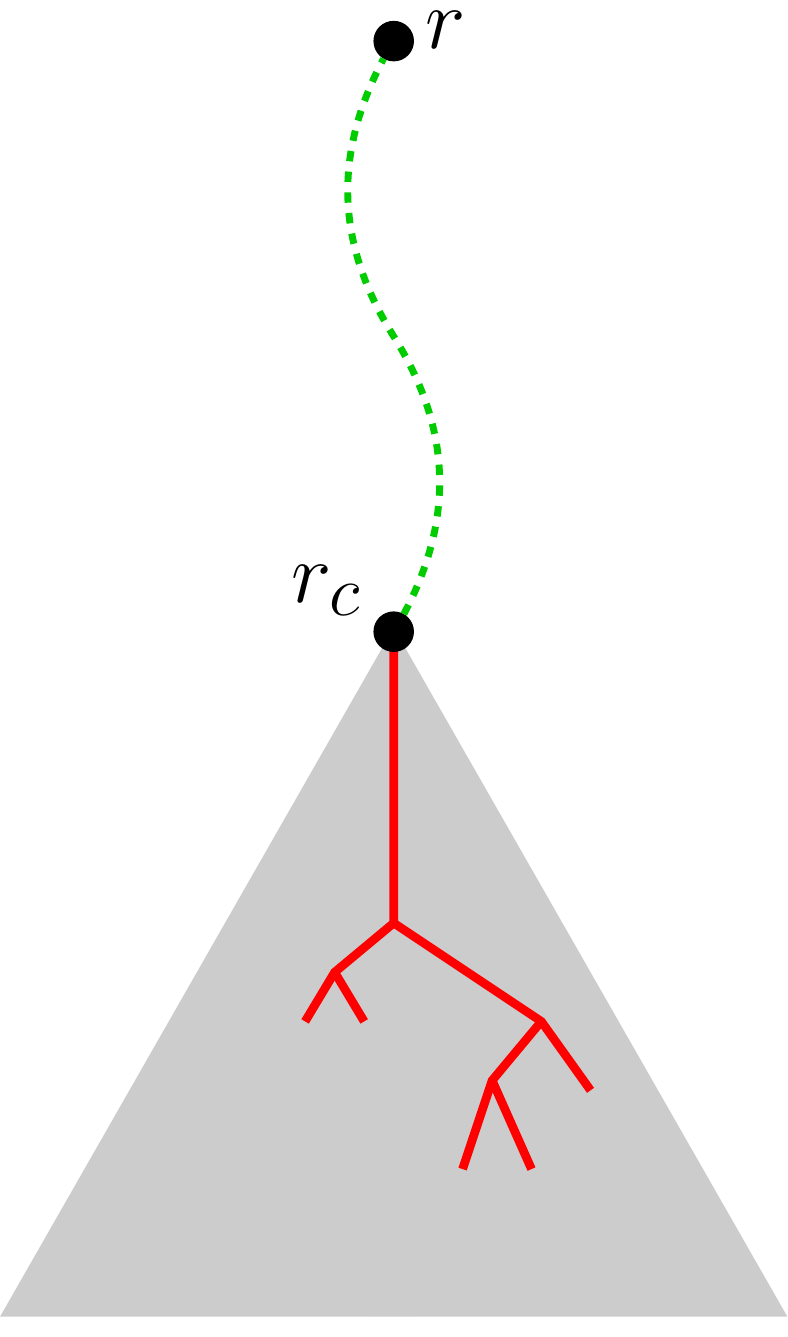}
\caption{Ending subtour $s$}
\label{fig:ending}
\end{subfigure}
\begin{subfigure}[t]{0.5\textwidth}
\centering
\includegraphics[scale=0.5]{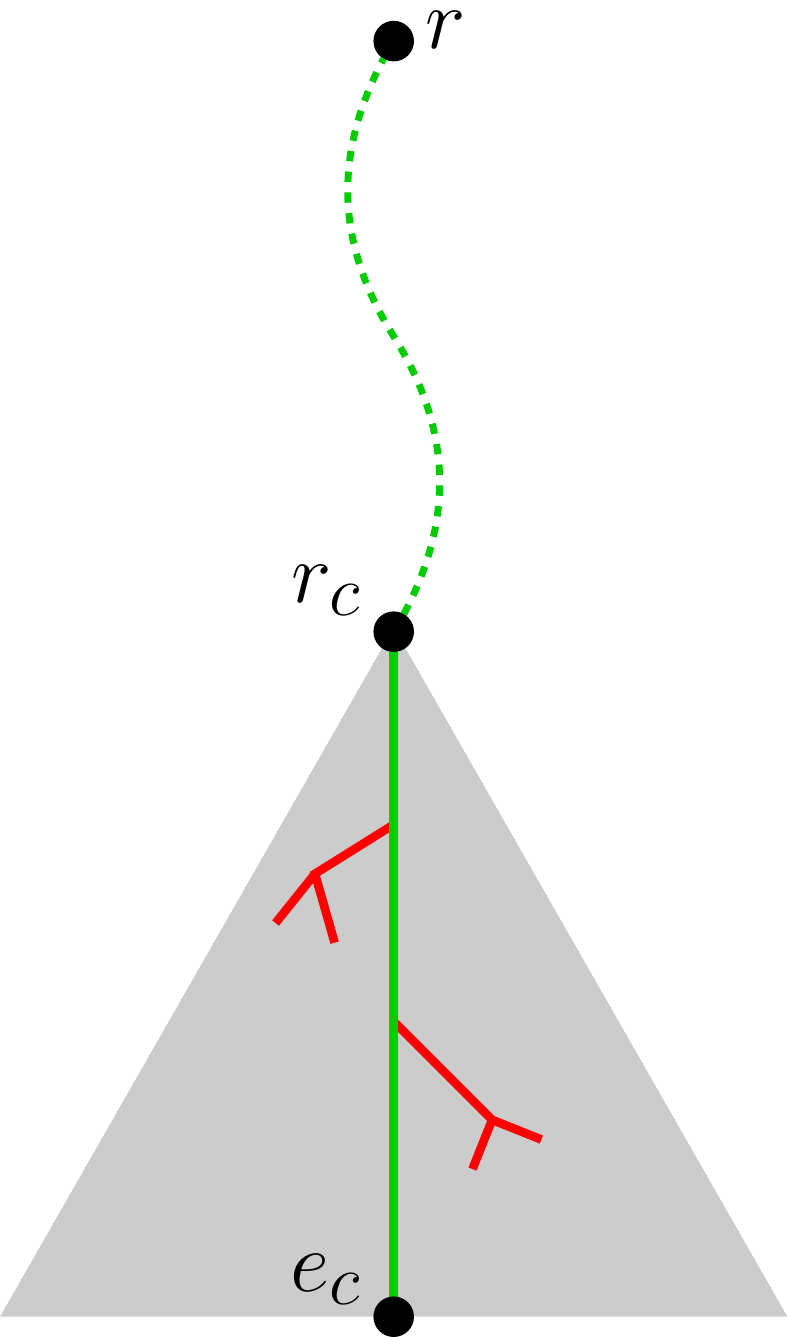}
\caption{Passing subtour $s$}
\label{fig:passing}
\end{subfigure}
\caption{Illustration for the new notion of the \emph{reduced length} of a subtour $s$ in a component $c$.
The component $c$ is represented by the gray triangle.
There are two cases depending on whether $s$ is an ending subtour (\cref{fig:ending}) or a passing subtour (\cref{fig:passing}).
In both cases, the subtour $s$ is represented by the solid segments.
The $r$-to-$r_c$ connection (in both directions) is represented by the dashed curve.
The reduced length $\bar{\ell}(s)$ of the subtour $s$ is defined by $\displaystyle\frac{\length(\text{red})}{D-\length(\text{green})}$.
See \cref{def:reduced-weight}.}
\label{fig:reduced_length}
\end{figure}

We distinguish the two categories of the subtours in the definition of reduced lengths (\cref{def:reduced-weight}) so that the properties in \cref{lemma:single-tour,thm:comp_sub} are satisfied.

%The distinction between ending and passing subtours is essential for the rest of this paper. Indeed, it is needed to prove the following lemma which draws a parallel between subtours in the DVRP problem and items in the bin-packing problem.

\begin{lemma}
\label{thm:comp_sub}
Let $t$ be a tour.
Let $c_1,\dots,c_k$ be the components containing terminals visited by $t$.
For each $i\in [1,k]$, let $s_i$ denote the subtour of $t$ in $c_i$.
%and  $\bar{\ell}_1, ..., \bar{\ell}_l$ be their reduced lengths.
Then
$\sum_{i=1}^k \bar{\ell}(s_i)\leq 1.$
\end{lemma}

\begin{proof}
Let $c^*\in\{c_1,\dots,c_k\}$ be a component such that $\dist(r,r_{c^*})$ is maximized.
Let $p$ be the path from $r$ to $r_{c^*}$.
For each $i\in[1,k]$, let $E_i\subseteq E$ denote a subset of edges defined as follows: if $s_i$ is an ending subtour, then $E_i$ consists of the edges in $s_i$; and if $s_i$ is a passing subtour, then $E_i$ consists of the edges in $s_i$ that do not belong to the $r_{c_i}$-to-$e_{c_i}$ path.
By construction, the edges in $p$, $E_1,\dots,E_k$ are disjoint.
Let $w(E_i)$ denote $\sum_{e\in E_i} w(e)$.
Since the length of $t$ is at most $D$, we have:
\[2\cdot \dist(r,r_{c^*}) + \sum_{i=1}^k 2\cdot w(E_i)  \leq D,\]
and equivalently,
\begin{equation}
\label{eqn:at-most-1}
\sum_{i=1}^k
\frac{2\cdot w(E_i)}{D - 2\cdot \dist(r,r_{c^*})} \leq 1.
\end{equation}
For each $i\in[1,k]$, by the definition of $E_i$, we have

\begin{equation}
\label{eqn:w_Ei}
    2\cdot w(E_i)=
\begin{cases}
\length(s_i), & \text{if $s_i$ is an ending subtour}\\
\length(s_i)-2\cdot \dist(r_{c_i},e_{c_i}), & \text{if $s_i$ is a passing subtour.}
\end{cases}
\end{equation}
By the choice of $c^*$, we have
\begin{equation}
\label{eqn:dist_r_rc*}
\dist(r,r_{c^*})\geq \begin{cases}
\dist(r,r_{c_i}), & \text{if $s_i$ is an ending subtour}\\
\dist(r,e_{c_i}), & \text{if $s_i$ is a passing subtour.}
\end{cases}
\end{equation}
From \cref{eqn:w_Ei,eqn:dist_r_rc*} and \cref{def:reduced-weight}, for all $i\in[1,k]$, we have
\[\bar{\ell}(s_i)\leq
\frac{2\cdot w(E_i)}{D-2\cdot \dist(r,r_{c^*})}.\]
Combining with \cref{eqn:at-most-1}, the claim follows.
\end{proof}

\begin{lemma}
\label{lemma:single-tour}
Let $c$ be any component.
Let $s_1,\dots,s_m$, for some $m\geq 1$, be any subtours in $c$ that are of the same category and such that $\sum_{i=1}^m\bar{\ell}(s_i) \leq 1.$
Then all terminals from all subtours $s_1,\dots,s_m$ can be visited by a tour of length at most $D$.
\end{lemma}

\begin{proof}
    We construct a tour $t$ visiting all terminals from all subtours $s_1,\dots, s_m$, and we show that its length is at most $D$.
    Since  $s_1,\dots,s_m$ are of the same category, there are two cases depending on their category.% subtour and bound its total length to prove that its length is at most $D$.

    \begin{itemize}
        \item If $s_1,\dots,s_m$ are ending subtours, then $t$ consists of the $r$-to-$r_c$ connection (in both directions) and of the subtour $s_i$ for each $i\in[1,m]$.
        Therefore,
        \begin{align*}
            \length(t) & \leq 2 \cdot \dist(r, r_c) + \sum_{i=1}^m \length(s_i) \\
            & = 2\cdot\dist(r, r_c) + (D - 2 \cdot \dist(r, r_c)) \sum_{i=1}^m \bar{\ell}(s_i) && \text{(by \cref{def:reduced-weight})} \\
            & \leq D && \text{(since $\sum_{i=1}^{m}\bar{\ell}(s_i)\leq 1$).}
        \end{align*}

        \item If $s_1,\dots,s_m$ are passing subtours, then $t$ consists of the $r$-to-$e_c$ connection (in both directions) and of the subtour $s_i$ without the  $r_c$-to-$e_c$ connection (in both directions) for each $i\in [1,m]$. Therefore,
        \begin{align*}
            \length(t) & \leq 2 \cdot \dist(r, e_c) + \sum_{i=1}^m \left( \length(s_i) - 2 \cdot \dist(r_c, e_c) \right) \\
            & = 2\cdot \dist(r, e_c) + (D - 2 \cdot \dist(r, e_c)) \sum_{i=1}^m \bar{\ell}(s_i) && \text{(by \cref{def:reduced-weight})} \\
            & \leq D && \text{(since $\sum_{i=1}^{m}\bar{\ell}(s_i)\leq 1$).}
        \end{align*}
    \end{itemize}
The claim follows in both cases.
\end{proof}

\begin{lemma}
\label{lem:number-components}
Assuming $\Gamma\geq 20$, the number of components in $\mathcal{C}$ is at most $(15 /\Gamma)\cdot \opt$.
\end{lemma}

\begin{proof}
Let $\mathcal{C}_b\subseteq \mathcal{C}$ denote the set of big components in $\mathcal{C}$.
Consider a big component $c\in \mathcal{C}_b$.
Let $s_{c,1}, ..., s_{c,m_c}$ denote all subtours in $c$ in the global optimal solution, for some $m_c\in \mathbb{N}$.
Let $\bar{\ell}_{c,1},\dots, \bar{\ell}_{c,m_c}$ denote their reduced lengths.
Those subtours can be  viewed as a feasible solution to the local problem for the terminals in $c$ only.
We partition those subtours into parts, such that the subtours in the same part are of the same category, and in addition, for each part except possibly two parts, the reduced lengths of the subtours in that part sum to a value in $(1/2,1]$. For each part, we replace the subtours by a new subtour visiting the terminals covered by all the subtours in that part. This creates a new feasible local solution for the terminals of $c$, and its number of subtours is at most $2\sum_{i=1}^{m_c} \bar{\ell}_{c,i} +2$.
By definition of big components, covering the terminals of $c$ requires at least  $ \Gamma/2$ tours. Thus:
 $$
 2\sum_{i=1}^{m_c} \bar{\ell}_{c,i} +2 \geq \Gamma/2.
 $$
Therefore,
 $$
 \sum_{i=1}^{m_c} \bar{\ell}_{c,i}\geq \Gamma/4 -1 \geq \Gamma/5,
 $$
since $\Gamma\geq 20$. Hence
$$
\sum_{c \in \mathcal{C}} \sum_{i=1}^{m_c} \bar{\ell}_{c,i}
\geq |\mathcal{C}_b| \cdot \Gamma /5 \geq
|\mathcal{C}| \cdot \Gamma / 15,
$$
where the last inequality follows from Property~4 of \cref{lem:comp_dec}.

On the other hand, consider each tour $t$ in $\OPT$. Let $k_t$ denote the number of components containing terminals visited by tour $t$.
For each $i \in [1,k_t]$, let $\bar{\ell}'_{t,i}$ denote the reduced length of the $i^{\rm th}$ subtour of $t$.
By \cref{thm:comp_sub} we have:
\begin{align*}
\sum_{t\in \OPT} \sum_{i=1}^{k_t} \bar{\ell}'_{t, i} \leq \sum_{t\in \OPT} 1 = \opt.
\end{align*}
By re-ordering this sum we get
$$\sum_{t\in \OPT} \sum_{i=1}^{k_t} \bar{\ell}'_{t, i}  = \sum_{c \in \mathcal{C}} \sum_{i=1}^{m_c} \bar{\ell}_{c,i}.$$
Therefore, $\opt \geq |\mathcal{C}| \cdot \Gamma / 15$. The claim follows.
\end{proof}

%\section{Analysis of the Algorithm}

%In this section, we show that \cref{alg:main} is a $1.691...$-approximation.
%\footnote{The proof of \cref{lem:comp_dec} already shows that it is a $2$-approximation.}

\section{Proof of \cref{thm:alpha-approx}}
\label{sec:proof-alpha-approx}

To prove \cref{thm:alpha-approx}, we  use a reduction from the DVRP on trees to the \emph{bounded space online bin-packing}.

\subsection{Bounded Space Online Bin Packing}
\label{sec:lee-lee}
%In the analysis of our algorithm, we  consider the \emph{bounded space online bin-packing} (\cref{def:online-packing}).
\begin{definition}
\label{def:online-packing}
In the \emph{bounded space online  bin-packing} problem, we are given a positive integer  $M$ and an online sequence of item sizes $(a_1, a_2, ..., a_n) \in [0,1]^n$, and we want to pack those items into  bins of size $1$.
At any time, the following operations are allowed: (1) opening a bin; (2) closing an bin; (3) assigning the current item to some \emph{open} bin.
We require that at any time there are at most $M$ open bins.
%there are at most $M$ {\em open} bins, i.e.,\ bins that are allowed to receive additional items in the future;
The goal is to minimize the total number of bins used.
\end{definition}

%The idea behind considering this variant of bin-packing is that by only having only a few bins open at any time, we can if necessary duplicate them at a minimal cost compared to the overall number of bins needed.

 Let $\opt_{\rm BP}$ denote the number of bins in an optimal solution to the bin packing in the \emph{offline} setting.
 The following theorem due to Lee and Lee \cite{leeleeonlinepacking} is a direct corollary of Theorem~2 from \cite{leeleeonlinepacking} and Eq.~(3.6) from \cite{leeleeonlinepacking}.

 \begin{theorem}[\cite{leeleeonlinepacking}]
Let $M$ be any positive integer.
Let $k\in\mathbb{N}$ be such that $u_k < M \leq u_{k+1}$.
There exists a solution to the bounded space online bin-packing
 %called HARMONIC$_M$,
in which the total number of bins used is at most
\[ \left(\sum_{i = 1}^k \frac{1}{u_i} + \frac{M}{(M-1)u_{k +1}}\right) \cdot \opt_{\rm BP}
+(M-1).\]
\label{theorem:leelee}
\end{theorem}

\subsection{Reduction}
In this subsection, we prove the following \cref{thm:reduction}.

\begin{theorem}
\label{thm:reduction}
\label{thm:alpha-construction}
Assume that the bounded space online bin-packing problem admits a solution using at most $\beta_M \cdot \opt_{\rm BP} + \gamma_M$ bins, where $\beta_M$ and $\gamma_M$ are  parameters depending on the space bound $M$, and $\opt_{\rm BP}$ is defined in \cref{sec:lee-lee}.
Let $\eps >0$.
Then \cref{alg:main} uses at most $(\beta_M + 30 \eps^2\cdot M) \cdot \opt + \gamma_M $ tours for the tree DVRP.
\end{theorem}

Consider an instance of the DVRP on a tree.
Let $m$ denote the number of components, and let $\mathcal C = \{ c_1, c_2, ..., c_m \}$ denote the set of components.
For each component $c_i\in \mathcal{C}$,  let $\bar{\ell}_{i,1}, \bar{\ell}_{i,2}, ..., \bar{\ell}_{i,e_i}$ denote the reduced lengths of the ending subtours in $c_i$ from $\OPT$, and $\bar{\ell}'_{i,1}, \bar{\ell}'_{i,2}, ..., \bar{\ell}'_{i,p_i}$ denote the reduced lengths of the passing subtours in $c_i$ from $\OPT$. We define the following instance of the bounded space online bin-packing:
\begin{align*}
L = (\bar{\ell}_{1,1}, \bar{\ell}_{1,2}, ..., \bar{\ell}_{1,e_1}, \bar{\ell}'_{1,1}, \bar{\ell}'_{1,2}, ..., \bar{\ell}'_{1,p_1}, ..., \bar{\ell}_{m,1}, \bar{\ell}_{m,2}, ..., \bar{\ell}_{m,e_m}, \bar{\ell}'_{m,1}, \bar{\ell}'_{m,2}, ..., \bar{\ell}'_{m,p_m}).
\end{align*}

Let $B_1$ denote a solution to this instance satisfying the assumption of the claim, i.e.,
\begin{equation}
\label{eqn:B1}
|B_1|\leq \beta_M\cdot \opt_{\rm BP} + \gamma_M.
\end{equation}

For each bin $b$ in $B_1$, we partition its contents so that each part contains the reduced lengths of the subtours that come from the same component and are in the same category, and we replace $b$ by a collection of bins, one for each part of the partition containing at least one subtour of $b$. This defines a bin-packing $B_2$.

Next, we define a solution $S$ to the tree DVRP corresponding to $B_2$.
For each bin $b$ of $B_2$, we consider the subtours that correspond to the reduced lengths in $b$, and we create one tour in $S$, which is the minimum tour visiting all terminals covered by those subtours.
Observe that those subtours are in the same component and of the same category and, in addition, their reduced lengths sum to at most $1$.
By \cref{lemma:single-tour}, the created tour is within the distance constraint $D$.
Therefore, $S$ is a feasible solution to the tree DVRP.

By construction, each tour in $S$ visits terminals from only one component.
Therefore, the output of \cref{alg:main} is at most $|S|$, so is at most $|B_2|$.
It remains to analyze $|B_2|$.
\begin{lemma}
\label{lemma:maj-b2}
$|B_2| \leq |B_1| + (30/\Gamma) M\cdot \opt.$
\end{lemma}

\begin{proof}
%We observe that in the sequence $L$ the subtours'  reduced lengths are given to the algorithm by grouping consecutively subtours in the same component and of the same category.
Let $J\subseteq[1,|L|-1]$ denote the set of integers $j\leq |L|-1$ such that the $j^{\rm th}$ element in $L$ and the $(j+1)^{\rm th}$ element in $L$ are the reduced lengths of two subtours in different components or of different categories.
From the construction of $L$ and since there are $m$ components and two categories, we have
\begin{equation}
\label{eqn:bound-J}
|J|\leq 2m.
\end{equation}

Consider a bin $b\in B_1$.
The number of bins in $B_2$ generated by $b$ is the number of pairs $(c,x)$ where $c\in \mathcal{C}$ and $x\in\{\text{passing, ending}\}$, such that $b$ contains a reduced length of a subtour in component $c$ and of category $x$.
Let $\min(b)$ (resp.\ $\max(b)$) denote the minimum (resp.\ maximum) integer $j\in [1,|L|]$ such that the $j^{\rm th}$ element in $L$ belongs to $b$.
Let $p_b$ denote the number of elements $j\in J$ such that $j\in[\min(b),\max(b)]$.
The number of bins in $B_2$ generated by $b$ is at most $1+p_b$.
Summing over all bins of $B_1$ gives
\begin{equation}
\label{eqn:B_2}
    |B_2|\leq \sum_{b \in B_1} (1 + p_b)=|B_1| + \sum_{b \in B_1} p_b.
\end{equation}
Observe that
\[\sum_{b \in B_1} p_b\leq \sum_{j\in J}
\#\big(\text{bins $b\in  B_1$ such that $j\in[\min(b),\max(b)]\big)$}.\]
For each $j\in J$, if a bin $b\in B_1$ is such that  $j\in[\min(b),\max(b)]$, then $b$ is \emph{open} at the time of $j$.
Since $B_1$ is a solution to the bounded space online bin packing, the number of open bins at the time of $j$ is at most $M$, thus the number of bins $b\in  B_1$ such that $j\in[\min(b),\max(b)]$ is at most $M$.
Therefore,
\begin{equation}
\label{eqn:p_b}
    \sum_{b \in B_1} p_b\leq |J|\cdot M\leq 2m\cdot M\leq (30/\Gamma)M\cdot\opt,
\end{equation}
where the second inequality follows from \cref{eqn:bound-J} and the last inequality follows from \cref{lem:number-components}.

The claim follows from \cref{eqn:B_2,eqn:p_b}.
\end{proof}

Combining \cref{lemma:maj-b2}, \cref{eqn:B1} and using $\Gamma=1/\eps^2$, we have
\[|B_2|\leq \beta_M \cdot\opt_{\rm BP} + \gamma_M + 30 \eps^2 M \cdot\opt.\]

Next, we show that $\opt_{\rm BP}$ is at most $\opt$.
For each tour $t$ in an optimal solution $\OPT$ to the tree DVRP, if $t$ visits components $c_1, ..., c_k$ with subtours' reduced lengths $\bar{\ell}_1, \bar{\ell}_2, ..., \bar{\ell}_k$, then using \cref{thm:comp_sub} we have $\bar{\ell}_1 + \bar{\ell}_2 + ... + \bar{\ell}_k \leq 1$. So it is possible to put the reduced lengths $\bar{\ell}_1, \bar{\ell}_2, ..., \bar{\ell}_k$ in a single bin of size $1$.
%Now we obtain a solution to the bin packing instance in the offline setting: for each tour $t \in \OPT$, put all the reduced lengths that correspond to a subtour of $t$ in a single bin.
This leads to a feasible solution to the bin packing in the offline setting that uses $\opt$ bins.
Thus $\opt_{\rm BP} \leq \opt.$
Hence
\[|B_2|\leq (\beta_M + 30 \eps^2 M )\cdot\opt + \gamma_M.\]
This completes the proof of \cref{thm:reduction}.

\subsection{Proof of \cref{thm:alpha-approx} Using the Reduction (\cref{thm:reduction})}

First, consider the case when $\opt<1/\eps^2$.
Since $\Gamma=1/\eps^2$,  \cref{alg:components} returns a single component, let it be $c$.
Then \cref{alg:local} computes an optimal solution in $c$, thus the solution returned by \cref{alg:main} is optimal.

Next, consider the case when $\opt\geq 1/\eps^2$.
Let $M=1/\eps$. Let $k$ be defined in \cref{theorem:leelee}.
By \cref{theorem:leelee}, there exists a solution to the bounded space online bin-packing problem using at most $\beta_M\cdot\opt_{\rm BP}+\gamma_M$ tours, where
\[\beta_M:=\left(\sum_{i = 1}^k \frac{1}{u_i} + \frac{M}{(M - 1)u_{k + 1}}\right)\leq \alpha + 2\eps, \text{ since } u_{k + 1} \geq M,\]
\[\gamma_M:=M-1< 1/\eps.\]
Thus by the reduction (\cref{thm:alpha-construction}), the number of tours used by \cref{alg:main} is at most \[(\beta_M + 30 \eps^2\cdot M) \cdot \opt + \gamma_M<(\alpha + 2\eps + 30 \eps)\cdot\opt + 1/\eps\leq (\alpha+33\eps)\cdot\opt,\] where the last inequality follows since $\opt\geq 1/\eps^2$.

By \cref{lem:comp_dec,lem:exact-algo}, the running time of \cref{alg:main} is $O\left(n^2\cdot \left( D / \eps^2 \right)^{O(1/\eps^2)}\right)$.

\section{Proof of \cref{thm:lower-bound}}
\label{sec:lower-bound}
For each positive integer $k$, we construct an instance $\mathcal{I}_k$ of the tree DVRP as follows.

The tree in the instance $\mathcal{I}_k$ consists of a root vertex $r$ and a set of components.
The root of each component is connected to $r$ by an edge of weight $0$.
%We will ensure that the minimum number of tours to cover the terminals in each component is $\Gamma$, so that the components correspond to the result of the decomposition in \cref{alg:components} (\cref{lem:comp_dec}).
The components are of $k$ \emph{types}.
For each $i\in[1,k]$, a type-$i$ component consists of $1+\Gamma \cdot u_i$ vertices: One vertex is the root of the component, and the remaining $\Gamma \cdot u_i$ vertices are the terminals, each connected to the root of the component by an edge of weight $x_i:=k \cdot u_{k+1} / (u_i + 1) + 1$.
There are $u_k / u_i$ components of type $i$ for each $i\in[1,k]$.
See \cref{fig:lower-bound}.
Observe that $u_{k+1} / (u_i + 1)$ and $u_k/u_i$ are both integers according to \cref{def:alpha}.
We set the distance constraint $D: = 2 k\cdot u_{k+1}$.

\begin{figure}
    \centering
    \includegraphics[scale=0.4]{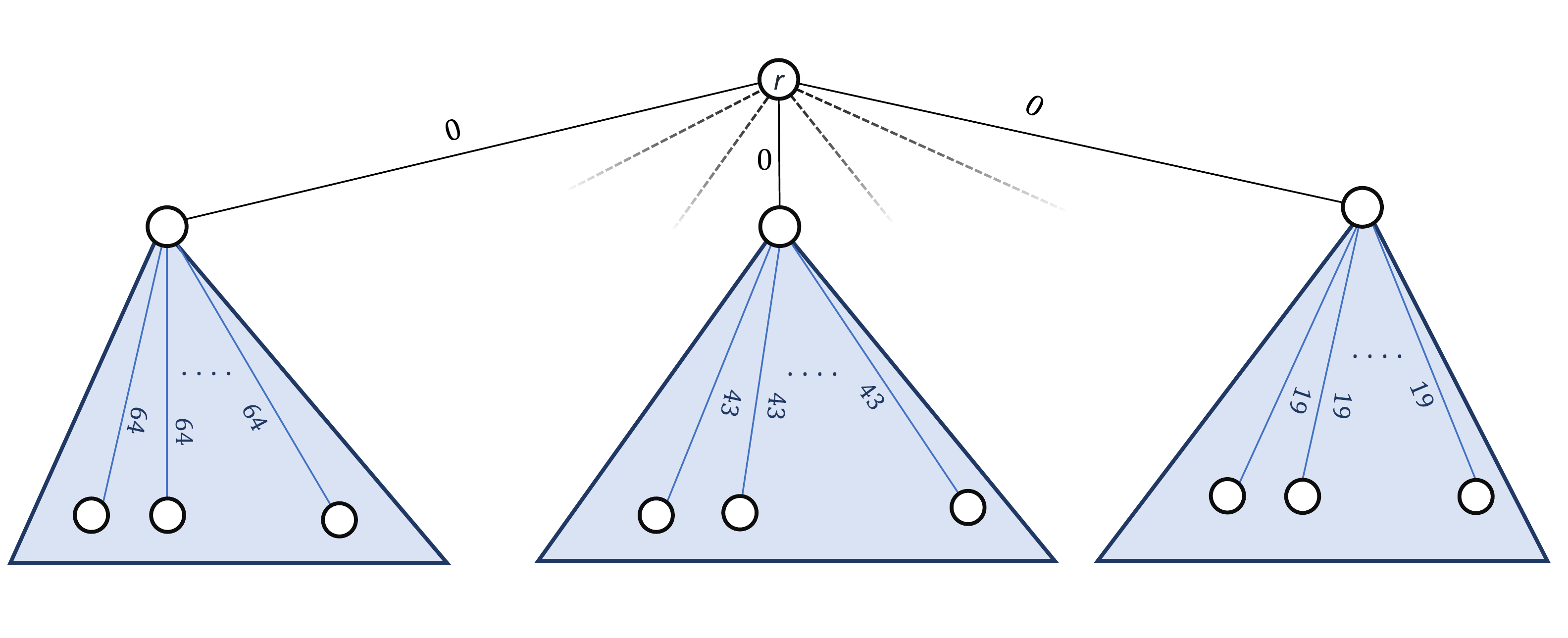}\hspace{2cm}
    \caption{Instance  $\mathcal{I}_k$ for $k$ = 3.
    From \cref{def:alpha}, $u_1 = 1$, $u_2=2$, $u_3=6$, $u_4=42$.
    The distance constraint $D= 2 \cdot k \cdot u_{k+1}  = 252$.
    In a type-1 component, the edge weight $x_1=64$.
    In a type-2 component, the edge weight $x_2=43$.
    In a type-3 component, the edge weight $x_3=19$.
    Consider a tour $t$ that visits a terminal in a type-1 component, a terminal in a type-2 component, and a terminal in a type-3 component.
    The length of $t$ is $2(x_1+x_2+x_3)= D$.
    %For the type-$2$ component, there are $u_3/u_2 = 6/2 = 3$ components, each with $\Gamma u_2 = 2 \Gamma$ vertices connected with an edge of weight $k u_4 / (u_2 + 1) + 1 = 3 \cdot 42 / 3 + 1 = 43$.
}
    \label{fig:lower-bound}
\end{figure}

We claim that there exists a feasible solution to $\mathcal{I}_k$ using $\Gamma \cdot u_k$ tours.
Consider a set of tours $S$ such that each tour $t\in S$ visits exactly $k$ terminals: for each $i\in[1,k]$, tour $t$ visits exactly one terminal from all components of type $i$.
For any $i \in [1, k]$, there are $\Gamma \cdot u_i \cdot u_k / u_i = \Gamma\cdot u_k$ terminals in all components of type $i$.
Thus $S$ consists of $\Gamma\cdot u_k$ tours.
Next, we show that each tour $t\in S$ is within the distance constraint $D$.
We have
\begin{equation}
\label{eqn:length-t}
    \length(t)=\sum_{i=1}^k 2 x_i  = 2 k \cdot u_{k+1} \sum_{i=1}^k \frac{1}{u_i +1} + 2k.
\end{equation}
Using \cref{def:alpha}, it is easy to show the following fact by induction.
\begin{fact}
\label{fact:induction}
For any positive integer $k$, we have
\[
    \frac{1}{u_{k+1}} = 1 - \sum_{i=1}^k \frac{1}{u_i + 1}.
\]
\end{fact}
From \cref{eqn:length-t,fact:induction}, we have \[\length(t) = 2 k \cdot u_{k+1} \left(1-\frac{1}{u_{k+1}}\right) + 2k  = D.\]
Therefore, $S$ is a feasible solution. Hence
\begin{equation}
\label{eqn:gamma_u_k}
    \opt \leq \Gamma \cdot u_k.
\end{equation}

Next, we analyze the solution to $\mathcal{I}_k$  computed by \cref{alg:main}.
Recall that \cref{alg:main} computes an optimal solution in each component independently.
Let $c$ be any component.
Let $i\in[1,k]$ be the type of $c$.
We observe that a tour is able to cover $u_i$ terminals in $c$ but is unable to cover $u_i + 1$ terminals in $c$.
This is because, the cost to cover $u_i$ terminals in $c$ is
\begin{align*}
    u_i  \cdot 2 x_i =  D -\frac{2k\cdot u_{k+1}}{u_{i}+1}+2u_i\leq D,
\end{align*}
where the inequality follows from \cref{def:alpha}, and the cost to cover $u_i+1$ terminals in $c$ is
\begin{align*}
    (u_i + 1) \cdot 2 x_i =  D + 2 (u_i + 1) > D.
\end{align*}
Since there are $\Gamma \cdot u_i$ terminals in $c$, the minimum number of tours to cover the terminals in $c$ is $\Gamma$.

For each $i\in [1,k]$, the number of components of type $i$ is $u_k/u_i$. Thus the number of tours returned by \cref{alg:main} is     $\sum_{i=1}^k (u_k / u_i) \cdot \Gamma  = \Gamma \cdot u_k \sum_{i=1}^k 1/u_i$.

Combined with \cref{eqn:gamma_u_k}, we conclude that the approximation ratio of \cref{alg:main} on $\mathcal{I}_k$ is at least $\sum_{i=1}^k 1/u_i$, which tends to $\alpha$ when $k$ tends to $\infty$ by \cref{def:alpha}.

This completes the proof of \cref{thm:lower-bound}.

\begin{comment}
\begin{algorithm}[h]
\caption{Explicit procedure to create from $\OPT$ an approximation where each tour picks terminals from only one component }
\label{alg:hat-T}
\begin{algorithmic}
\For{each component $c \in C$}
	Make $1/\eps$ bins
	\For{$i \in [1,..,1/\eps - 1]$}
		\State Put in bin $i$ the subtours whose reduced lengths is in $[1/i, 1/(i+1)[$
	\EndFor
	\State Put in the last bin the remaining subtours (of length at most $\eps$
	\For{each bin $i$}
		\While{there is at least a tour left in bin $i$}
			\State Create a new tour that takes greedily as many subtours as possible in bin $i$
		\EndWhile
	\EndFor
\EndFor
\State Return the set containing all the tours that were created.
\end{algorithmic}
\end{algorithm}
\end{comment}

%\newpage

\bibliography{references}

%\newpage

\appendix

%\section*{Appendix}

\section{NP-Hardness}
\label{sec:NP-hard}
\begin{lemma}
The Tree DVRP is strongly NP-hard.
\end{lemma}

\begin{proof}
We reduce the  bin packing problem to the tree DVRP.
Consider an instance of the bin packing problem with an  integer bin capacity $M$ and $n$ items of sizes $a_1,\dots,a_n$, where $a_i\in [0,M]$ for each $i\in[1,n]$.
The bin packing problem looks for a partition of the $n$ items into the minimum number of bins such that in each bin, the sum of the item sizes is at most $M$.
We construct an instance of the tree DVRP as follows.
There is a depot and $n$ terminals.
For each $i\in[1,n]$, there is an edge between the depot and the $i^{\rm th}$ terminal with weight $a_i$.
Let $D = 2M$.
It is easy to see that a solution to the bin packing instance is equivalent to a soution to the tree DVRP instance.
Since the bin packing problem is strongly NP-hard~\cite{garey_computers_1985}, the tree DVRP is strongly NP-hard.
\end{proof}

\section{Component Decomposition Algorithm}
\label{sec:alg-component}
For each vertex $v$ of the tree, let $T(v)$ denote the subtree rooted at $v$.
For any subgraph $H$ of the tree, let $U_H$ denote the set of terminals in $U$ that belong to the subgraph $H$.

The algorithm is given in \cref{alg:components}.

\begin{algorithm}[h]
\caption{Decomposition algorithm \textsc{Decompose}$^{(\Gamma)}$ parameterized by $\Gamma$ (see \cref{lem:comp_dec})}
\label{alg:components}
\hspace*{\algorithmicindent} \textbf{Input} A tree $T$ rooted at $r$, a distance constraint $D$ \\
\hspace*{\algorithmicindent} \textbf{Output} A decomposition of $T$ into components
\begin{algorithmic}[1]
\State $\{$Leaf components$\}:=\{T(v): v$ least deep vertex s.t. \textsc{Solve}$^{(\Gamma)}(T(v),D-2\cdot\dist(r,v),U_{T(v)}) \leq \Gamma$
%{SOLVE($T(v)$) can be covered by at most 2\Gamma tours}\}$
\State $T'\gets$ subtree spanning $\{ r\}\cup\{$roots of leaf components$\}$
\For{each maximal downward $v_1$-to-$v_2$ path in $T'$ whose internal vertices  have only one child in $T'$}
    \State Let $v'_1$ be the child of $v_1$ on the $v_1$-to-$v_2$ path.
    \While{$\textsc{Solve}^{(\Gamma)}(T(v)\setminus T(v_2),D-2\cdot\dist(r,v),U_{T(v)\setminus T(v_2)}) \leq \Gamma$}
        \State $v\gets$ least deep vertex on the $v_1'$-to-$v_2$ path such that \\ \hspace{45 pt} $\textsc{Solve}^{(\Gamma)}(T(v)\setminus T(v_2),D-2\cdot\dist(r,v),U_{T(v)\setminus T(v_2)}) \leq \Gamma$
        \State Define internal component $(T(v)\setminus T(v_2))$ with
        exit vertex $v_2$
        \State $v_2\gets v$
    \EndWhile
        \State Define internal component  $(v_1,v'_1)\cup (T(v'_1)\setminus T(v_2))$ with
        exit vertex~$v_2$
\EndFor
\end{algorithmic}
\end{algorithm}

\paragraph*{Running time}
The number of calls on $\textsc{Solve}^{(\Gamma)}$ is $O(n)$.
Each call takes time $O(n\cdot (2\Gamma)^{\Gamma}\cdot D^{3\Gamma})$ by \cref{lem:exact-algo}.
Thus the overall running time of \cref{alg:components} is $O(n^2\cdot (2\Gamma)^{\Gamma}\cdot D^{3\Gamma})$.

\section{Proof of \cref{lem:exact-algo}}
\label{sec:proof-exact-algo}
Let $\tilde T=(\tilde V,\tilde E)$ be a tree with root $\tilde r$.
Let $\tilde n$ denote the number of vertices in $\tilde T$.
Let $\tilde U\subseteq \tilde V$ be the set of terminals.
The goal is to cover the terminals in $\tilde U$ using a minimum number of tours of length at most $\tilde D$ each.

Let $\Gamma$ be a positive integer.
We design a dynamic program that computes an optimal solution of at most $\Gamma$ tours if such a solution exists.
See \cref{alg:local}.

\begin{algorithm}
\caption{Dynamic program \textsc{Solve}$^{(\Gamma)}$ parameterized by $\Gamma$ (see \cref{lem:exact-algo})}
\label{alg:local}
\hspace*{\algorithmicindent} \textbf{Input} A tree $\tilde T$ rooted at $\tilde r$, a set of terminals $\tilde U$, a distance constraint $\tilde D$\\
\hspace*{\algorithmicindent} \textbf{Output}
$\min \{ |S|: S \text{ is a feasible set of tours and }|S|\leq \Gamma \}$
\begin{algorithmic}[1]
\State Preprocess the tree $\tilde T$ so that each leaf is a terminal and each non-leaf vertex of $\tilde T$ has exactly two children \Comment{\cref{sec:problem}}
\For{each configuration $(v, A)$}
    \State $valid(v,A)=false$
\EndFor
    \For{each terminal $v$ in $\tilde T$}\Comment{Case 1}
    \For{each list $A$ such that $\ell(A)\in [1,\Gamma]$ and every element in $A$ equals 0}
    \State $valid(v,A)=true$
    \EndFor
    \For{each non-terminal $v$ of $\tilde T$ in bottom-up order }\Comment{Case 2}
        \State Let $v_1$ and $v_2$ denote the two children of $v$
        \For{each lists $A, A_1, A_2$ such that $(v_1,A_1)$ and $(v_2,A_2)$ are valid configurations}
	        \If{
        $(v_1,A_1)$, $(v_2,A_2)$, and $(v,A)$ are \emph{compatible}} \Comment{\cref{def:compatible}}
	        \State $valid(v,A)\gets true$
            \EndIf
        \EndFor
    \EndFor
    \EndFor
\State \Return $\displaystyle \min \{ \ell(A) : (\tilde r,A) \ \text{is valid} \}$
\end{algorithmic}
\end{algorithm}

To begin with, using the preprocessing step in \cref{sec:problem}, we transform the tree $\tilde T$ so that the terminals are the same as the leaves in $\tilde T$ and every non-leaf vertex in $\tilde T$ has exactly two children.

We compute values at \emph{configurations} (\cref{def:local-config}), which are solutions restricted to a subtree of $\tilde T$.
%Recall that, by \cref{lem:comp_dec}, for any component, the maximum number of tours to cover the terminals in that component is at most $\Gamma$.

\begin{definition}[configurations]
\label{def:local-config}
%Let $S\in \mathcal{C}$ be a subtree rooted at $r_s$.
%Let $D_S \geq 0$, $D_S \leq D$.
A \emph{configuration} $(v,A)$ is defined by a vertex $v\in \tilde T$ and a list $A$ of $\ell(A)$ integers
$(s_1,s_2,\ldots ,s_{\ell(A)})$ such that
\begin{itemize}
\item
$\ell(A)\leq \biggamma$;
\item
for each $i\in[1, \ell(A)]$,  $s_i$ is an integer in $[0,\tilde D]$.
\end{itemize}
%When $v=r_c$, the local configuration $(r_c,A)$ is also called a \emph{local configuration in the component $c$}.
\end{definition}

We say that a configuration $(v,A)$ is \emph{valid} if it is possible to cover the terminals in the subtree of $\tilde T$ rooted at $v$ with a collection of $\ell(A)$ subtours, such that each subtour starts and ends at $v$ and the $i$-th subtour has length $s_i$.
Note that ${\ell(A)}$ is equal to the total number  of subtours in the collection. Thus the objective is to find a valid configuration $(\tilde r,A)$ with $\ell(A)$ minimum.

%If $c$ has an exit vertex $e_c$, we first modify $c$ as follows: we simply remove $e_c$; and if there now is an internal vertex that has only one child in $c$, then we eliminate it by contracting the two edges incident to it into a single edge with weight equal to the sum of the two weights. From now on we assume that there is no exit vertex.

Let $v$ be any vertex in $\tilde T$.
We decide whether the configuration $(v,A)$ is valid according to one of the two cases.

\subsubsection*{Case 1: $v$ is a leaf vertex  in $\tilde T$}
Then $v$ is a terminal. The configuration $(v,A)$ is \emph{valid} if and only if $\ell(A)\in [1,\Gamma]$ and every element in the list $A$ equals 0.

\subsubsection*{Case 2: $v$ is a non-leaf vertex in $\tilde T$}
Let $v_1$ and $v_2$ be the two children of $v$ in $\tilde T$.
%We assume without loss of generality that a non-terminal $v$ has two children in $c$. Indeed, if $v$ has no child in $c$ (which happens when $v$ is an exit vertex of $c$), then we simply remove $v$; and if $v$ has only one child in $c$, then we eliminate $v$ by contracting the two edges incident to $v$ into a single edge \cm{with weight equal to the sum of the two weights}.

\begin{definition}[compatibility]
\label{def:compatible}
Let $w_1$ (resp.\ $w_2$) denote the weight of the edge between $v$ and $v_1$ (resp.\ $v_2$).
We say that the configurations $(v_1,A_1)$, $(v_2,A_2)$, and $(v,A)$ are \emph{compatible} if there is a partition $\mathcal{P}$ of $A_1\cup A_2$ into parts, each part consisting of one or two elements such that at most one element is from $A_1$ (resp.\ $A_2$), and a one-to-one correspondence between every part in $\mathcal{P}$ and every element in $A$ such that:
\begin{itemize}
    \item a part in $\mathcal{P}$ consisting of one element $s^{(1)}\in A_1$ corresponds to an element $s$ in $A$ if and only if $s^{(1)}+2w_1=s$;
    \item a part in $\mathcal{P}$ consisting of one element $s^{(2)}\in A_2$ corresponds to an element $s$ in $A$ if and only if $s^{(2)}+2w_2=s$;
    \item a part in $\mathcal{P}$ consisting of two elements $s^{(1)}$ and $s^{(2)}$ corresponds to an element $s$ in $A$ if and only if $s=s^{(1)}+2w_1+s^{(2)}+2w_2$.
\end{itemize}
\end{definition}

The configuration $(v,A)$ is \emph{valid} if and only if there exist a valid configuration $(v_1,A_1)$ and a valid configuration $(v_2,A_2)$ that are compatible with $(v,A)$.

%The algorithm is very simple.

\paragraph*{Running time}
For each vertex $v\in \tilde T$, since $\ell(A)\leq \Gamma$, the number of configurations $(v,A)$ is at most $\tilde n\cdot \tilde D^{\Gamma}$.
%$n^{O_\eps(1)}$.
For fixed $(v_1,A_1),(v_2, A_2)$, and $(v,A)$, to check  compatibility, there are at most $(2\Gamma)^{\Gamma}$ partitions of $A_1\cup A_2$ into parts.
Thus the overall running time of \cref{lem:exact-algo} is $O(\tilde n\cdot (2\Gamma)^{\Gamma}\cdot  \tilde D^{3\Gamma})$.

\end{document}